\documentclass[runningheads]{llncs}
\usepackage{graphicx}
\usepackage{amsmath,amsfonts,amssymb,color}
\usepackage{fancybox}

\usepackage{tikz}

\usepackage[ruled,vlined]{algorithm2e}
\usepackage{algorithmic}
\usepackage{float}

\def\defeq{\stackrel{\mathrm{def}}{=}}

\newenvironment{fminipage}%
{\begin{Sbox}\begin{minipage}}%
		{\end{minipage}\end{Sbox}\fbox{\TheSbox}}

\begin{document}
\title{Optimal Offline Dynamic $2,3$-Edge/Vertex Connectivity}
\author{Richard Peng\inst{1} \and
Bryce Sandlund\inst{2} \and
Daniel D. Sleator\inst{3}}
\authorrunning{Peng et al.}
\institute{School of Computer Science, Georgia Institute of Technology, Atlanta, GA, USA
\email{rpeng@cc.gatech.edu} \and
Cheriton School of Computer Science, University of Waterloo, Waterloo, ON, CA
\email{bcsandlund@uwaterloo.ca} \and
Dept. of Computer Science, Carnegie Mellon University, Pittsburgh, PA, USA
\email{sleator@cs.cmu.edu}}
\maketitle              % typeset the header of the contribution
%
%!TEX root = Main.tex

\begin{abstract}
We give offline algorithms for processing a sequence of $2$ and $3$ edge and vertex connectivity queries in a fully-dynamic undirected graph. While the current best fully-dynamic online data structures for $3$-edge and $3$-vertex connectivity require $O(n^{2/3})$ and $O(n)$ time per update, respectively, our per-operation cost is only $O(\log n)$, optimal due to the dynamic connectivity lower bound of Patrascu and Demaine. Our approach utilizes a divide and conquer scheme that transforms a graph into smaller equivalents that preserve connectivity information. This construction of equivalents is closely-related to the development of vertex sparsifiers, and shares important connections to several upcoming results in dynamic graph data structures, outside of just the offline model.
\end{abstract}

%!TEX root = Main.tex

\section{Introduction}
\label{sec:intro}

Dynamic graph data structures seek to answer queries on a graph as it undergoes edge insertions and deletions. Perhaps the simplest and most fundamental query to consider is connectivity.
A connectivity query asks for the existence of a path
connecting two vertices $u$ and $v$ in the current graph.
As the insertion or deletion of a single edge may have large consequences to connectivity across the entire graph, constructing an efficient dynamic data structure to answer connectivity queries has been a challenge to the data structure community. A number of solutions have been developed, achieving a wide variety of runtime tradeoffs in a number of different
models~\cite{Eppstein94,Eppstein97,Eppstein06,KapronKM13,HolmLT98,HRT18,Huang17,KaczmarzL15,LackiS13,Thorup00,Wulff-Nilsen17}

The models typically addressed are \emph{online}: each query must be answered before the next query or update is given. A less demanding variant is the \emph{offline} setting, where the entire sequence of updates and queries is provided as input to the algorithm.
While an online data structure is more general, there are many scenarios in which the entire sequence of operations is known in advance. This is often the case when a data structure is used in a subroutine of an algorithm~\cite{LiPYZ18:arxiv,Bringmann18}, one specific example being the use of dynamic trees in the near-linear time minimum cut algorithm of Karger~\cite{Karger08}.

In exchange for the loss of flexibility, one can hope to obtain faster and simpler algorithms in the offline setting. This has been shown to be the case in the dynamic minimum spanning tree problem. While an online fully-dynamic minimum spanning tree data structure requires about $O(\log^4 n)$ time per update~\cite{Holm15}, the offline algorithm of Eppstein requires only $O(\log n)$ time per update~\cite{Eppstein94}.

In this paper, we show similar, but stronger, performance gains for higher versions of connectivity. In particular, we consider the problems of $2,3$-edge/vertex connectivity on a fully-dynamic undirected graph. An extension of connectivity, $c$-edge connectivity asks for the existence of $c$ edge-disjoint paths between two vertices $u$ and $v$ in the current graph. Vertex connectivity requires vertex-disjoint paths instead of edge-disjoint paths. Current online fully-dynamic $2$-edge/vertex connectivity data structures require update time at least $O(\log^2 n)$~\cite{HRT18} and $O(\log^3 n)$\footnote{This complexity is claimed in Thorup's STOC 2000~\cite{Thorup00} result. As noted by Huang et al.~\cite{Huang17}, the paper provides few details, deferring to a journal version that has since not appeared. The best complexity for online fully-dynamic biconnectivity prior to this claim was $O(\log^5 n)$ by Holm and Thorup~\cite{HolmLT98}.}~\cite{Thorup00}, respectively, and current online fully-dynamic $3$-edge/vertex connectivity data structures require update time at least $O(n^{2/3})$ and $O(n)$, respectively~\cite{Eppstein97}. In contrast, our offline algorithms for $2,3$-edge/vertex connectivity require only $O(\log n)$ time per operation. As the lower bound on dynamic connectivity~\cite{DemaineP04}, as well as most lower bounds in general~\cite{AbboudWY15,Abboud14,Abboud16,Dahlgaard16}, also apply in the offline model, our algorithms are optimal up to a constant factor. This paper further shows that any lower bound attempting to show hardness stronger than $\Omega(\log n)$ time per operation for online fully-dynamic $2,3$-edge/vertex connectivity must make use of the online model.

As a straightforward application of our results, one can consider the use of our algorithms when data regarding a dynamic network is collected, but not analyzed, until a later point in time. For example, to diagnose an issue of network latency across key routing hubs, or determine viability of a dynamically-changing network of roads, our algorithms can answer a batch of queries in time $O(t \log n)$, where $t$ is the total number of updates and queries. Since online fully-dynamic algorithms for higher versions of connectivity are significantly slower, namely, $O(n^{2/3})$ and $O(n)$ time update for $3$-edge and $3$-vertex connectivity, respectively, our offline data structure makes these computations practical for large data sets when they would otherwise be prohibitively expensive.

Related to our work are papers by Lacki and Sankowski~\cite{LackiS13} and Karcmarz and Lacki~\cite{KaczmarzL15}, which also apply to the above applications but for lower versions of connectivity. Their work considers a fixed sequence of graph updates, given in advance, and is then able to answer connectivity queries regarding intervals of this sequence, online. This is more general than the model we consider because the queries need not be supplied in advance and data regarding an interval of time is richer than information from a specific point in time. For connectivity and $2$-edge connectivity, Karcmarz and Lacki achieve $O(\log n)$ time per operation~\cite{KaczmarzL15}. Both $2$-vertex connectivity and $3$-edge/vertex connectivity queries are not supported.

The techniques developed in this paper may be of independent interest. Our work has close connections with recent developments in vertex sparsification, particularly vertex sparsification-based dynamic graph data structures~\cite{AbrahamDKKP16,DurfeeKPRS17,GoranciHP17,GoranciHP18:arxiv,Durfee19,Eppstein06,Goranci18,Kratsch12,Assadi2015DynamicSF,Fafianie16,Fafianie16b}. In particular, the equivalent graphs at the core of our algorithms are akin to vertex sparsifiers, with the main difference that $2$- and $3$-connectivity require preserving far less information than the more general definitions of vertex sparsifiers~\cite{GoranciHP17,Kratsch12}. A promising step in this direction is very recent work of Goranci et al.~\cite{Goranci18}, which suggests the notion of a \emph{local sparsifier}. This is a generalization of the sparsifier that we consider here, and leads to efficient incremental algorithms in the \emph{online} setting.

Indeed, offline algorithms haven proven useful for the development of online counterparts in the past. Once such example is recent development in the maintenance of dynamic effective resistance. Recent work in fully-dynamic data structures for maintaining effective resistances online~\cite{Durfee19} relied heavily upon ideas from earlier data structures for maintaining effective resistances in offline~\cite{DurfeeKPRS17,LiZ18} or offline-online hybrid~\cite{DurfeeKPRS17,LiPYZ18:arxiv} settings.

The results of this paper were previously published online in the open access journal arXiv~\cite{PengSS} and have recently been extended to offline $4$- and $5$-edge connectivity~\cite{Molina19}. This new work achieves about $O(\sqrt{n})$ time complexity per operation.

The rest of this paper will be dedicated to proving the following theorem:

\begin{theorem}
\label{thm:main}
Given a sequence of $t$ updates/queries on a graph of the form:
\begin{itemize}
	\item Insert edge $(u, v)$,
	\item Delete edge $(u, v)$,
	\item Query if a pair of vertices $u$ and $v$ are
$2$-edge connected/$3$-edge connected/bi-connected/tri-connected in the current graph,
\end{itemize}
there exists an algorithm that answers all queries in $O(t\log{n})$ time.
\end{theorem}

For simplicity, we will assume the graph is empty at the start and end of the sequence, but the results discussed are easily modified to start with an initial graph $G$, at the cost of an additive $O(m)$ term in the running time, where $m$ is the number of edges of $G$. Further, we assume a fixed vertex set of size $n$. Any update sequence with arbitrary vertex endpoints can be modified to one on a fixed set of vertices, where the size of the fixed set is equal to the largest number of non-isolated vertices in any graph achieved in the given update sequence.

The paper is organized as follows. We describe our offline framework for reducing graphs to smaller equivalents in Section~\ref{sec:framework}. We show how simple techniques can be used to create such equivalents for $2$-edge connectivity and bi-connectivity in Section~\ref{sec:Two}. In Section~\ref{sec:ThreeEdge} we extend these constructions to $3$-edge connectivity. Our most technical section is $3$-vertex connectivity, where constructing equivalent graphs requires careful manipulation of SPQR trees. Unfortunately, due to page limits, this is located in Appendix~\ref{sec:ThreeVertex}.

%!TEX root = Main.tex

\section{Offline Framework}
\label{sec:framework}
The main idea of our offline framework is to perform divide and conquer on the input sequence, similar to what is done in Eppstein's offline minimum spanning tree algorithm~\cite{Eppstein94}. Consider the full sequence of updates and queries $x_1, \ldots, x_t$, where each $x_i$ is either an edge insertion, edge deletion, or query. Call each $x_i$ an event.

Assume each inserted edge has unique identity. Then for each inserted edge $e$, we may associate an interval $[I(e), D(e)]$, indicating that edge $e$ was inserted at time $I(e)$ and removed at time $D(e)$.
%If an edge $e$ is present in the initial graph, let $I(e) = 1$, and if edge $e$ is present at the end of the update sequence, let $D(e) = t$.
Plotting time along the $x$-axis and edges on the $y$-axis as in Figure \ref{fig:timelinefull} gives a convenient way to view the sequence of events.

\begin{figure}[h]
\begin{center}

\begin{tikzpicture}[x=1cm, y=1cm]
\usetikzlibrary{arrows,positioning} 

\tikzset{
    arrow/.style={
           |->,
           thick},
    empty/.style={white}
}

\node at (0,2.5) {$e_4$}; 
\node at (0,2) {$e_3$}; 
\node at (0,1.5) {$e_2$}; 
\node at (0,1) {$e_1$};
\draw (7,2.5) edge[arrow] (10,2.5); 
\draw (4,2) edge[arrow] (8,2); 
\draw (2,1.5) edge[arrow] (5,1.5); 
\draw (1,1) edge[arrow] (11,1); 
\draw (1,.5) edge[dotted] (1,2.7);
\node at (1, 0) {I($e_1$)};
\draw (2,.5) edge[dotted] (2,2.7);
\node at (2, 0) {I($e_2$)};
\draw (3,.5) edge[dotted] (3,2.7);
\node at (3, 0) {Q};
\draw (4, .5) edge[dotted] (4,2.7);
\node at (4, 0) {I($e_3$)};
\draw (5,.5) edge[dotted] (5,2.7);
\node at (5, 0) {D($e_2$)};
\draw (6,.5) edge[dotted] (6,2.7);
\node at (6, 0) {Q};
\draw (7, .5) edge[dotted] (7,2.7);
\node at (7, 0) {I($e_4)$};
\draw (8, .5) edge[dotted] (8,2.7);
\node at (8, 0) {D($e_3$)};
\draw (9, .5) edge[dotted] (9,2.7);
\node at (9, 0) {Q};
\draw (10, .5) edge[dotted] (10,2.7);
\node at (10, 0) {D($e_4$)};
\draw (11, .5) edge[dotted] (11,2.7);
\node at (11, 0) {D($e_1$)};
\node at (1, 3) {1};
\node at (2, 3) {2};
\node at (3, 3) {3};
\node at (4, 3) {4};
\node at (5, 3) {5};
\node at (6, 3) {6};
\node at (7, 3) {7};
\node at (8, 3) {8};
\node at (9, 3) {9};
\node at (10, 3) {10};
\node at (11, 3) {11};

\end{tikzpicture}
\end{center}

\caption{A timeline diagram of four edge
insertions(I)/deletions(D) and three queries(Q), with time on the $x$-axis and edges on the $y$-axis.}
\label{fig:timelinefull}
\end{figure}
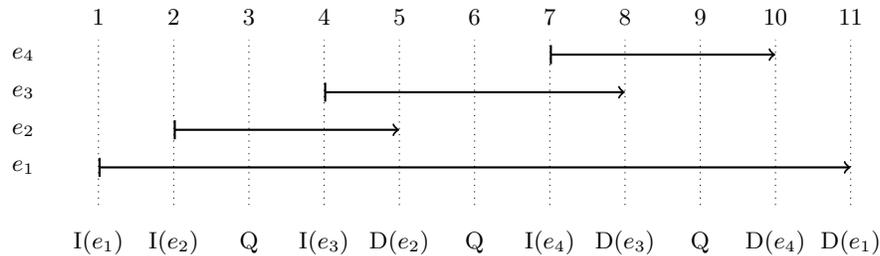

Fix some subinterval $[l,r]$ of the sequence of events. Let us classify all edges present at any point of time in the sequence $x_l, \ldots, x_r$ as one of two types.
\begin{enumerate}
\item Edges present throughout the duration ($i_e \leq l \leq r \leq d_e$),
we call \emph{permanent} edges.
\item Edges affected by an event in this range (one or both of $I(e), D(e)$ is in $(l, r)$),
we call \emph{non-permanent} edges.
\end{enumerate}

While there may be a large number of permanent edges,
the number of non-permanent edges is limited by the number
of time steps, $r - l + 1$.
Therefore, the graph can be viewed as a large static
graph on which a smaller number of events take place.

Our goal will be to reduce this graph of permanent edges to one
whose size is a small function of the number of events in the subinterval.
If we may do so without affecting the answers to the queries,
we can recursively apply the technique to achieve an efficient divide and conquer
algorithm for the original dynamic $c$-connectivity sequence.

We will work in the dual-view, considering cuts instead of edge-disjoint or vertex-disjoint paths. Two vertices $u$ and $v$ are $c$-edge connected if there does not exist a cut of size $c-1$ edges separating them; further, $u$ and $v$ are $c$-vertex connected if there does not exist a cut of $c-1$ vertices that separates them.

We need the following definition.

\begin{definition}
\label{equivalence}
Given a graph $G=(V_G,E_G)$ with vertex subset $W \subseteq V_G$ and a graph $H=(V_H,E_H)$ with $W \subseteq V_H$, we say that $H$ and $G$ are $c$-edge equivalent if, for any partition $(A,B)$ of $W$, the size of a minimum cut separating $A$ and $B$ is the same in $G$ and $H$ whenever either of these sizes is less than $c$. Similarly, we say $H$ and $G$ are $c$-vertex equivalent if, for any partition $(A,B,C)$ of $W$ with $|C| < c$, the size of a minimum vertex cut $D$ separating $A$ and $B$ such that $C \subseteq D$ and $D \cap A = \emptyset$, $D \cap B = \emptyset$, is the same in $G$ and $H$ whenever either of these sizes is less than $c$.
\end{definition}

%Suppose $G$ and $H$ are equivalent for $c$-edge connectivity or $c$-vertex connectivity with vertex set $W$.
This gives the following.

\begin{lemma}
\label{lem:AddEdges}
Suppose $G=(V_G,E_G)$ and $H=(V_H,E_H)$ are $c$-edge/$c$-vertex equivalent on vertex set $W$. Let $E_W$ denote any set of edges between vertices of $W$. Then $H'=(V_H,E_H \cup E_W)$ and $G'=(V_G,E_G \cup E_W)$ are $c$-edge/$c$-vertex equivalent.
\end{lemma}
\begin{proof}
We first show $c$-edge equivalence.
Let $(A, B)$ be any partition of $W$ and consider the minimum cuts separating $A$ and $B$ in $G'$ and $H'$.
Since the edges in $E_W$ are between vertices of $W$, they must cross the separation $(A, B)$ in the same way.
Therefore, if the minimum cut separating $A$ and $B$ had size less than $c$ in either $G$ or $H$, the minimum cuts
separating $A$ and $B$ will have equivalent size in $G'$ and $H'$. Further, if the minimum cuts separating $A$ and $B$ had
size larger or equal to $c$ in both $G$ and $H$, the minimum cuts separating $A$ and $B$ will also have size larger or equal to $c$ in $G'$ and $H'$, since we only add edges to $G'$ and $H'$. Thus $G'$ and $H'$ are $c$-edge equivalent.

We now consider $c$-vertex equivalence. Consider a partition $(A, B, C)$ of $W$. As with edge connectivity, if no vertex subset $D$ exists satisfying the conditions of Definition~\ref{equivalence}, the introduction of additional edges between any vertices of $W$ will not change the existence of such a set $D$ in $G'$ or $H'$. Furthermore, if an edge of $E_W$ connects a vertex of $A$ to a vertex of $B$, no vertex cut separates $A$ and $B$ in $G'$ and $H'$. Now suppose none of these cases is true, and there exists a vertex set $D$ satisfying the conditions of Definition~\ref{equivalence} such that the removal of $D$ disconnects $A$ and $B$ in $G$ and $H$ and further that no edge of $E_W$ connects a vertex of $A$ to a vertex of $B$. Then the removal of vertex set $D$ still disconnects $A$ and $B$ in $G'$ and $H'$. Thus $c$-vertex equivalence of $G'$ and $H'$ follows from $c$-vertex equivalence of $G$ and $H$.
\end{proof}

Now consider the graph $G$ of permanent edges for the subinterval $x_l, \ldots, x_r$ of events. Let $W$ be the set of vertices involved in any event in the subinterval. We will refer to these vertices as \emph{active} vertices, and all other vertices of $G$ not in $W$ as \emph{inactive} vertices. Lemma \ref{lem:AddEdges} says that if we reduce $G$ to a $c$-edge/ $c$-vertex equivalent graph $H$ on set $W$, the result of all queries in $x_l, \ldots, x_r$ on $H$ will be the same as on $G$. This is because all cuts in $H$ and $G$ that affect the queries (therefore of size less than $c$) are of equivalent size, even after the addition of non-permanent edges in $H$ and $G$.

This idea can lead to a divide and conquer algorithm if we can produce such equivalent graphs $H$ of small size efficiently. Specifically:

\begin{lemma}
\label{lem:framework}
Given a graph $G$ with $m$ edges and vertex set $W$ of size $k$, if there is an $O(m)$ time algorithm that produces a graph $H$ of size $O(k)$ that is $c$-edge/$c$-vertex equivalent to $G$ on $W$, then there is an algorithm that can answer all $c$-edge/$c$-vertex connectivity queries in a sequence of events $x_1, \ldots, x_t$ in $O(t \log n)$ time.
\end{lemma}
\begin{proof}
We perform divide and conquer on the sequence of events. We take the sequence of events $x_1, \ldots, x_t$ and divide it in half. Over each half, we will take the graph of permanent edges, which we denote $G$, and reduce it to a $c$-edge/ $c$-vertex equivalent graph $H$. We repeat the scheme recursively. As the subintervals get smaller, non-permanent edges become permanent and are absorbed by the production of equivalent graphs. Eventually, we reduce to subintervals with a constant number of events, which can be answered by any algorithm of our choice on a graph of constant size.

Consider the sizes of the graphs in each step of recursion. At the top level, the graph of permanent edges in each half has size $O(t)$, the number of events in the full sequence. Thus in the second level, the graph we start with is of size linear in the number of events in that interval, plus the number of edges that became permanent in that interval, which in turn are bounded by the number of non-permanent edges in the interval at the preceding level. This continues recursively, and the result is that all graphs operated on are of size linear in the number of events processed at that level. The runtime recurrence is $T(t) = T(t/2) + O(t)$, which clearly solves to $O(t \log t)$. If $t$ is polynomially-bounded by $n$, this is $O(t \log n)$. If not, we may first break the sequence of events $x_1, \ldots x_t$ into blocks of size, say, $n^2$. Since the size of the graph cannot be more than $O(n^2)$ in any block, the runtime argument carries through. This proves the lemma.
\end{proof}

The remainder of the paper will show the construction of $2$-edge, $2$-vertex, $3$-edge, and $3$-vertex equivalent graphs.

%!TEX root = Main.tex

\section{Equivalent Graphs for $2$-Edge Connectivity and Bi-connectivity}
\label{sec:Two}

We now show offline algorithms for dynamic $2$-edge connectivity and
bi-connectivity by constructing $2$-edge and $2$-vertex equivalents
needed by Lemma \ref{lem:framework}.
These two properties ask for the existence of a single edge/vertex whose
removal separates query vertices $u$ and $v$.
Since these cuts can affect at most one connected component,
it suffices to handle each component separately.

The underlying structure for $2$-edge connectivity and bi-connectivity
is tree-like.
This is perhaps more evident for $2$-edge connectivity, where vertices
on the same cycle belong to the same $2$-edge connected component.
We will first describe the reductions that we will make to this tree
in Section \ref{subsec:twoedge}, and adapt them to
bi-connectivity in Section \ref{subsec:twovert}.

At times we will make use of the term ``equivalent cut". By this we mean that
a cut $C'$ is equivalent to $C$ if it has the same size and separates the vertices 
of $W$ in the same way.

\subsection{$2$-Edge Connectivity}
\label{subsec:twoedge}

Using depth-first search~\cite{HopcroftT73}, we can identify all
cut-edges in the graph and the $2$-edge connected components
that they partition the graph into.
The case of edge cuts is slightly simpler conceptually,
since we can combine vertices without introducing new cuts.
Specifically, we show that each $2$-edge connected
component can be shrunk to single vertex.

\begin{lemma}
\label{lem:twocomp}
Let $S$ be a $2$-edge connected component in $G$.
Then contracting all vertices in $S$ to a single vertex $s$ in $H$\footnote{Here we abuse our requirement $W \subseteq V_H$, where $V_H$ are the vertices of $H$, slightly. A map of $W$ onto $V_H$ that preserves the cuts needed by $c$-edge/$c$-vertex equivalence is all that is required, but $c$-edge/$c$-vertex equivalence is easier to read without introducing the map notation.},
and endpoints of edges correspondingly, creates a
$2$-edge equivalent graph.
\end{lemma}

\begin{proof}
The only cuts that we need to consider are ones that
remove cut edges in $G$ or $H$.
Since we only contracted vertices in a component,
there is a one-to-one mapping of these edges from
$G$ to $H$.
Since $S$ is $2$-edge connected, all vertices in it
will be on the same side of one of these cuts.
Furthermore, removing the same edge in $H$
leads to a cut with $s$ instead.
Therefore, all active vertices in $S$ are mapped
to $s$, and are therefore on the same side of the cut.
\end{proof}

This allows us to reduce $G$ to a tree $H$, but the size of this
tree can be much larger than $k$.
Therefore we need to prune the tree by removing
inactive leaves and length $2$ paths whose
middle vertex is inactive.

\begin{lemma}
\label{lem:twoedgetreeprune}
If $G$ is a tree, the following two operations lead
to $2$-edge equivalent graphs $H$.
\begin{itemize}
	\item Removing an inactive leaf.
	\item Removing an inactive vertex with degree $2$ and adding an edge between its two neighbors.
\end{itemize}
\end{lemma}

\begin{proof}
In the first case, the only cut in $G$ that no longer
exist in $H$ is the one that removes the cut edge
connecting the leaf with its unique neighbor.
However, this places all active vertices in one component
and thus does not separate $W$ and need not
be represented in $H$.

In the second case, if a cut removes either of the edges incident to
the degree $2$ vertex, removing the new edge creates an equivalent
cut since the middle vertex is inactive.
Also, for a cut that removes the new edge in $H$, removing either of
the two original edges in $G$ leads to an equivalent cut.
\end{proof}

This allows us to bound the size of the tree by the number of
active vertices, and therefore finish the construction.

\begin{lemma}
\label{lem:twoequivalent}
Given a graph $G$ with $m$ edges and $k$ active vertices $W$,
a $2$-edge equivalent of $G$ of size $O(k)$, $H$, can be
constructed in $O(m)$ time.
\end{lemma}

\begin{proof}
We can find all the cut edges and $2$-edge connected components
in $O(m)$ time using depth-first search \cite{HopcroftT73}, and
reduce the resulting structure to a tree $H$ using Lemma \ref{lem:twocomp}.
On $H$, we repeatedly apply Lemma \ref{lem:twoedgetreeprune}
to obtain $H'$.

In $H'$, all leaves are active, and any inactive internal vertex
has degree at least $3$.
Therefore the number of such vertices can be bounded by $O(k)$,
giving a total size of $O(k)$.
\end{proof}

\subsection{Bi-connectivity}
\label{subsec:twovert}

All cut-vertices (articulation points) can also be identified
using DFS, leading to a structure known as the block-tree.
However, several modifications are needed to adapt the
ideas from Section \ref{subsec:twoedge}.
The main difference is that we can no longer replace each
bi-connected component with a single vertex in $H$, since cutting
such vertices corresponds to cutting a much larger set in $G$.
Instead, we will need to replace the bi-connected components
with simpler bi-connected graphs such as cycles.

\begin{lemma}
\label{lem:bicomp}
Replacing a bi-connected component with a cycle containing
all its cut-vertices and active vertices gives a $2$-vertex equivalent graph.
\end{lemma}

\begin{proof}
As this mapping maintains the bi-connectivity of the component,
it does not introduce any new cut-vertices.
Therefore, $G$ and $H$ have the same set of cut vertices
and the same block-tree structure.
Note that the actual order the active vertices appear in
does not matter, since they will never be separated.
The claim follows similarly to Lemma \ref{lem:twocomp}.
\end{proof}

The block-tree also needs to be shrunk in a similar manner.
Note that the fact that blocks are connected by shared vertices
along with Lemma \ref{lem:bicomp} implies the removal of
inactive leaves.
Any leaf component with no active vertices aside from
its cut vertex can be reduced to the cut vertex, and
therefore be removed.
%\todo{This is a bit of a peculiarity of block trees, is there more that can be said?}
The following is an equivalent of the degree two removal
part of Lemma \ref{lem:twoedgetreeprune}.

\begin{lemma}
\label{lem:bitreeprune}
Two bi-connected components $C_1$ and $C_2$ with no active
vertices that share cut vertex $w$ and are only incident to one other cut
vertex each, $u$ and $v$ respectively, can be replaced by an
edge connecting $u$ and $v$ to create a $2$-vertex
equivalent graph.
\end{lemma}

\begin{proof}
As we have removed only $w$, any cut vertex in $H$ is
also a cut vertex in $G$.
As $C_1$ and $C_2$ contain no active vertices,
this cut would induce the same partition of active vertices.

For the cut given by removing $w$ in $G$, removing $u$
in $H$ gives the same cut since $C_1$ has no active vertices
(which in turn implies that $u$ is not active).
Note that the removal of $u$ may break the graph into more
pieces, but our definition of cuts allows us to place these
pieces on two sides of the cut arbitrarily.
\end{proof}

Note that Lemma \ref{lem:bicomp} may need to be applied
iteratively with Lemma \ref{lem:bitreeprune} since some of the
cut vertices may no longer be cut vertices due to the removal
of components attached to them.

\begin{lemma}
\label{lem:biequivalent}
Given a graph $G$ with $m$ edges and $k$ active vertices $W$,
a $2$-vertex equivalent of $G$ of size $O(k)$, $H$, can be
constructed in $O(m)$ time.
\end{lemma}

\begin{proof}
We can find all the initial block-trees using depth-first
search \cite{HopcroftT73}.
Then we can apply Lemmas \ref{lem:bicomp} and \ref{lem:bitreeprune}
repeatedly until no more reductions are possible.
Several additional observations are needed to run these
reduction steps in $O(m)$ time.
As each cut vertex is removed at most once, we can keep a counter
in each component about the number of cut vertices on it.
Also, the second time we run Lemma \ref{lem:bicomp} on a component,
it's already a cycle, so the reductions can be done without examining
the entire cycle by tracking it in a doubly linked list and
removing vertices from it.

It remains to bound the size of the final block-tree.
Each leaf in the block-tree has at least one active vertex
that's not its cut vertex.
Therefore, the block-tree contains at most $O(k)$
leaves and therefore at most $O(k)$ internal components
with $3$ or more cut vertices, as well as $O(k)$ components
containing active vertices.
If these components are connected by paths with $4$
or more blocks in the block tree,
then the two middle blocks on this path meet the condition of
Lemma \ref{lem:biequivalent} and should have been
removed by Lemma \ref{lem:bitreeprune}.
This gives a bound of $O(k)$ on the number of blocks, which in
turn implies an $O(k)$ bound on the number of cut vertices.
The edge count then follows from the fact that
Lemma \ref{lem:bicomp} replaces each component with a cycle,
whose number of edges is linear in the number of vertices.
\end{proof}

%!TEX root = Main.tex

\section{$3$-Edge Connectivity}
\label{sec:ThreeEdge}

We now extend our algorithms to $3$-edge connectivity.
Our starting point is a statement similar to
Lemma~\ref{lem:twocomp}, namely that we can contract
all $3$-edge connected components.

\begin{lemma}
\label{lem:ThreeComp}
Let $S$ be a $3$-edge connected component in $G$.
Then contracting all vertices in $S$ to a single vertex $s$ in $H$,
and endpoints of edges correspondingly, creates a
$3$-edge equivalent graph.
\end{lemma}

\begin{proof}
A two-edge cut will not separate a $3$-edge connected component.
Therefore all active vertices in $S$ fall on one side of the cut, to which
vertex $s$ may also fall. The proof follows analogously to Lemma~\ref{lem:twocomp}.
\end{proof}

Such components can also be identified in $O(m)$ time
using depth-first search~\cite{Tsin09},
so the preprocessing part of this algorithm is the same
as with the 2-connectivity cases.
However, the graph after this shrinking step is no longer
a tree.
Instead, it is a cactus, which in its simplest terms can be
defined as:
\begin{definition}
\label{def:cactus}
A cactus is an undirected graph where each edge belongs to at
most one cycle.
\end{definition}

On the other hand, cactuses can also be viewed as a tree
with some of the vertices turned into cycles\footnote{Some 
`virtual' edges are needed in this construction,
because a vertex can still belong to multiple cycles.}.
Such a structure essentially allows us to repeat the same operations
as in Section~\ref{sec:Two} after applying the initial contractions.

\begin{lemma}
\label{lem:Remove3EdgeCon}
A connected undirected graph with no nontrivial $3$-edge connected component
is a cactus.
\end{lemma}

Due to space restrictions, we save the proof for Appendix~\ref{sec:Appendix}.

With this structural statement, we can then repeat
the reductions from the $2$-edge equivalent
algorithm from Section~\ref{subsec:twoedge}
to produce the $3$-edge equivalent graph.

\begin{lemma}
\label{lem:ThreeEdge}
Given a graph $G$ with $m$ edges and $k$ active vertices $W$,
a $3$-edge equivalent of $G$ of size $O(k)$, $H$, can be
constructed in $O(m)$ time.
\end{lemma}

\begin{proof}
Lemma~\ref{lem:Remove3EdgeCon} means that we can reduce
the graph to a cactus after $O(m)$ time preprocessing.

First consider the tree where the cycles are viewed as vertices.
Note that in this view, a vertex that's not on any cycle is also
viewed as a cycle of size $1$.
This can be pruned in a manner analogous to Lemma~\ref{lem:twoedgetreeprune}:
\begin{enumerate}
\item Cycles containing no active vertices and incident to
$1$ or $2$ other cycles can be contracted to a single vertex.
\item Inactive single-vertex cycles incident to $1$ other cycle can be removed.
\end{enumerate}
This procedure takes $O(m)$ time and produces a graph with
at most $O(k)$ leaves. Correctness of the first rule follows by replacing a cut of the
two edges within an inactive cycle by a cut of the single contracted
vertex with one of its neighbors. The second rule does not affect any cuts separating $W$.
It remains to reduce the length of degree $2$ paths and the sizes
of the cycles themselves.

As in Lemma~\ref{lem:twoedgetreeprune}, 
all vertices of degree
$2$ can be replaced by an edge between its two neighbors.
This bounds the length of degree $2$ paths and
reduces the size of each cycle to at most
twice its number of incidences with other cycles.
This latter number is in turn bounded by the number of
leaves of the tree of cycles.
Hence, this contraction procedure reduces the total size to $O(k)$.
\end{proof}

We remark that this is not identical to iteratively removing inactive
vertices of degrees at most $2$.
With that rule, a cycle can lead to a duplicate edge between pairs of vertices,
and a chain of such cycles needs to be reduced in length.

\bibliographystyle{splncs04}
\bibliography{../ref}

\appendix

%!TEX root = Main.tex

\section{Tri-connectivity}
\label{sec:ThreeVertex}

Tri-connectivity queries involving $s$ and $t$ ask
for the existence of a separation pair $\{u, v\}$ whose
removal disconnects $s$ from $t$.
In this section we will extend our techniques to offline tri-connectivity.
To do so, we rely on a tree-like structure for the set of separation
pairs in a bi-connected graph, the SPQR tree \cite{HopcroftT73,DiBattistaT90}.
We review these structures in Section \ref{subsec:spqr} and 
show how to trim them in Section \ref{subsec:spqrtrim}.
As these trees require that the graphs are bi-connected, we extend
this subroutine to our full construction in Section \ref{subsec:trifull}.

\subsection{SPQR Trees}
\label{subsec:spqr}

We now review the definition of SPQR Trees. We will follow the model
given in Chapter 2 of \cite{Weiskircher02}. For a more thorough description of SPQR
Trees, please refer to \cite{Weiskircher02}.

SPQR trees are based on the definition of a split pair,
which generalizes separation pairs by allowing the extra
case of $\{u, v\}$ being an edge of $G$.
A split component of the split pair $\{u, v\}$ is either an edge
connecting them, or a maximal connected subgraph $G'$
of $G$ such that removing $\{u, v\}$ does not disconnect $G'$.

The SPQR tree is defined
recursively on a graph $G$ with a special split pair $\{s, t\}$.
This process can be started by picking an arbitrary edge as the root.
Each node $\mu$ in the tree $\mathcal{T}$ has an associated graph denoted
as its skeleton, $skeleton(\mu)$, and is associated with an edge
in the skeleton of its parent $\nu$, called the virtual edge of $\mu$
in $skeleton(\nu)$. In this way, each virtual edge of a node $\nu$ corresponds
to a child of $\nu$.
For contrast, we will also use real edges to denote edges that are
present in $G$.

\begin{itemize}
\item Trivial Case: if $G$ is a single edge from $s$ to $t$, then $\mathcal{T}$
contains a single Q-node whose skeleton is $G$ itself.
\item Series Case: If the removal of the (virtual) edge $st$ creates cut-vertices,
then these cut vertices partition $G$ into blocks $G_1 \ldots G_k$
and the block-tree has a cycle-like structure.
The root of $\mathcal{T}$ is then an S-node, and $skeleton(\mu)$ is the
cycle containing these cut vertices with virtual edges corresponding
to the blocks.
\item Parallel Case: If the split pair $\{s, t\}$ creates split components
$G_1 \ldots G_k$ with $k \geq 2$, the root of $\mathcal{T}$ is a P-node
and the skeleton contains $k$ parallel virtual edges from $s$ to $t$ corresponding
to the split components.
\item Rigid Case: If none of the above cases apply, then the root of $\mathcal{T}$
is an R-node $\mu$ and $skeleton(\mu)$ is a tri-connected graph where
each edge corresponds to a split pair $\{s_i, t_i\}$, and the corresponding child
contains the union of all split components generated by this pair.
\end{itemize}

Note that by this construction, each edge in the SPQR tree
corresponds to a split-pair.
An additional detail that we omitted is that the construction of R-nodes
picks only the split pairs that are maximal w.r.t. the edge $st$. This detail
is unimportant to our trimming routine, but is discussed in \cite{Weiskircher02}.
Our algorithm acts directly upon the SPQR tree,
and we make use of several important properties of
this tree in our reduction routines.

When viewed as an unrooted tree, the SPQR tree
is unique with all leaves as Q-nodes.
For simplicity, we will refer to this tree with all Q-nodes
removed as the simplified SPQR tree, or $\mathcal{T}_{simple}$.
Also, it suffices to work on the skeletons of nodes,
and the only candidates for  separation pairs that we need to
consider are:
\begin{enumerate}
	\item Two cut vertices in an S-node.
	\item The split pair corresponding to a P-node.
	\item Endpoints of an edge in an R-node.
\end{enumerate}

\subsection{Trimming SPQR trees}
\label{subsec:spqrtrim}

We now show how to convert a bi-connected graph with $k$
active vertices to a $3$-vertex equivalent with $O(k)$
vertices and edges.
Our algorithm makes a sequence of modifications on the SPQR
trees similar to the trimming from Section \ref{sec:Two}.
We will call a vertex $u$ \textit{internal} to a node $\mu$ of the SPQR tree if $u$ is contained
in the split graph of $\mu$, and $u$ is
not part of the split pair associated with $\mu$.
We call a vertex $u$ \textit{exact} to a node $\mu$ if it is internal to $\mu$ but not to any
descendants of $\mu$.
 Exact vertices provide a way to count nodes according to
active vertices without reusing active vertices for multiple nodes.

We first give a way to remove split components with no internal active vertices,
which we will refer to as inactive split components.

\begin{lemma}
\label{lem:hangingcomp2}
Consider an inactive non-Q split component $G'$ produced by the split pair $\{u, v\}$.
We may create a $3$-vertex equivalent graph $H$ by the following replacement rule:
if $u$ and $v$ are $\geq3$-vertex connected in $G$, we may replace $G'$ with the edge $uv$;
otherwise, we may replace $G'$ with a vertex $x$ and edges $ux$, $vx$.
\end{lemma}

\begin{proof}
Consider a cut in $G$. If $u$ and $v$ are on the same side of the cut,
an equivalent cut exists in $H$ by placing $u$, $v$, and possibly $x$ on the same side of the cut.
Similarly, a cut in $G$ that removes either $u$ or $v$ can be made in $H$ by removing the
same vertex. If a cut in $G$ separates $u$ and $v$, $u$ and $v$ are not 
$\geq3$-vertex connected in $G$
and a vertex $w$ in $G'$ must have been removed
since $G'$ connects $u$ and $v$. Removing $x$ instead of $w$ creates an equivalent cut in $H$.

In the other direction, if $u$ and $v$ were $\geq3$-vertex connected in $G$, also no cut in $H$
separates $u$ and $v$. If $u$ and $v$ were not $\geq3$-vertex connected in $G$, there exists a
vertex $w$ in $G'$ whose removal separates $u$ and $v$ in $G'$. 
Therefore the cut in $H$ produced by vertices $\{x, z\}$ for some $z \notin G'$
that separates $u$ and $v$ can be made in $G$ by removing $\{w, z\}$ instead.
All other cuts in $H$ can be formed in $G$ by placing $G'$ on the same
side as any remaining vertex in $\{u, v\}$.
\end{proof}

An important consequence of Lemma \ref{lem:hangingcomp2} is that an inactive
split component is a leaf in $\mathcal{T}_{simple}$.
Because of this, it will be easier to think of inactive split components in the same way
we think of Q-nodes. We define the tree $\mathcal{T}_{simple'}$ as the tree $\mathcal{T}_{simple}$
with inactive split component leaves removed. Every leaf $\mu$ in $\mathcal{T}_{simple'}$ must have
an internal active vertex $u$. Furthermore, any vertex internal to a node $\eta$
 is only internal to ancestors and possibly descendants 
of $\eta$. Since $\mu$ is a leaf in $\mathcal{T}_{simple'}$, it has no descendants. It follows that $u$ is
exact for $\mu$ and the tree $\mathcal{T}_{simple'}$ has $O(k)$ leaves.

We will use this to bound the number of leaves in $\mathcal{T}$.
However, if $G$ contains an R-node whose skeleton contains a complete graph between active
vertices, the number of leaves in $\mathcal{T}$ can still be $\Omega(k^2)$.
Therefore, we need another rule to process each skeleton of $\mathcal{T}_{simple'}$ so that we can bound
its size by the number of its exact active vertices and split components.

\begin{lemma}
\label{lem:nodereduction}
There exists a constant $c_0$ such that any node $\mu$ in $\mathcal{T}_{simple'}$
whose skeleton contains a total of $k$ exact active vertices and virtual
edges corresponding to active split components can be replaced by
a node whose skeleton contains $c_0 \cdot k$ vertices/edges to give
a $3$-vertex equivalent graph.

In other words, the number of children of $\mu$ in $T$ is at most $c_0 \cdot k$.
\end{lemma}

\begin{proof}
We consider the cases where the node is of type S, P, R separately.

If the node is a P-node, it suffices to consider the case where
there are $3$ or more inactive split components incident
to the cut pair.
Removing at most $2$ vertices from these components
will leave at least one of them intact, and therefore not change
the connectivity between the separation pair, and therefore
the other components.
Therefore, all except $3$ of these components can be discarded.

If the node is an S-node, any virtual edge in the cycle corresponding
to an inactive split component that's not incident to two active
vertices is replaced by actual edges via Lemma \ref{lem:hangingcomp2}.
Furthermore, two consecutive real edges in the cycle connecting
three inactive vertices can be reduced to a single edge by
removing the middle vertex in a manner analogous to
Lemma \ref{lem:twoedgetreeprune}.

Therefore the size of the cycle is proportional to the number of exact active
vertices and active split components.

For an R-node, we can identify all vertices that are either exact and active,
or are incident to virtual edges corresponding to active split
components. 
Then we can simply connect these vertices together using a
tri-connected graph of linear size (e.g. a wheel graph)
and add the active split components back between their respective separation pairs.
Every cut in the new skeleton with a child in $\mathcal{T}_{simple}$ can be produced
by the original graph,
and the only cuts in the original graph that can't be produced
in the new one are cuts isolating an inactive component.
Therefore the two graphs are $3$-vertex equivalent.
\end{proof}

Before we can bound the number of nodes in $\mathcal{T}_{simple'}$, we must eliminate
long paths where a node has only one active child, which in turn has only
one active child, etc. This is only possible if there exists an active vertex (vertices) internal to each split
component, otherwise by Lemma~\ref{lem:hangingcomp2}, a node with no internal active vertices
becomes a leaf. Since the internal active vertex (vertices) are shared amongst all ancestors, there is no
way to associate the active vertex with a constant number of SPQR nodes. The replacement rule
in this Lemma provides a way to get around this.

\begin{lemma}
\label{lem:tritreeprune}
Let $\mu_1$, $\mu_2$, and $\mu_3$ be a parent-child sequence of SPQR nodes where $\mu_2$ and
$\mu_3$ are associated with 
 split pairs $\{u, v\}$ and $\{x, y\}$ respectively,
 $\mu_3$ is the single active child of $\mu_2$,
 $\mu_2$ is the single active child of $\mu_1$, and either none of $u$, $v$, $x$, and $y$ are active
 or $u = x$ and is active.
 We may replace $\mu_2$ with $\mu_3$, effectively replacing vertex $u$ with $x$ and $v$ with $y$.
\end{lemma}

\begin{proof}
In $G$, cuts induced by the removal of $\{u, v\}$ and $\{x, y\}$ both give the same partition of active vertices.
Therefore, we only need one in $H$, which is given by the cut $\{x, y\}$. All other cuts in $G$ not present
in $H$ only separate inactive split components, which need not be represented in $H$.

In the other direction, every cut in $H$ is still present in $G$.
Thus the new graph $H$ preserves $3$-vertex equivalence.
\end{proof}

We can now bound the number of nodes in $\mathcal{T}_{simple'}$ after no reductions
via Lemmas~\ref{lem:hangingcomp2},~\ref{lem:nodereduction}, and \ref{lem:tritreeprune} are possible.

\begin{lemma}
\label{lem:manyevent}
Consider a graph $G$ where no reductions via
Lemmas~\ref{lem:hangingcomp2}, \ref{lem:nodereduction}, and \ref{lem:tritreeprune}
are possible. There is a constant $c_1$ such that the number of nodes in $G$'s
SPQR tree $\mathcal{T}$ is no more than  $c_1 \cdot k$, where $k$ is the total number of
active vertices in $G$.
\end{lemma}

\begin{proof}

As explained earlier, the tree $\mathcal{T}_{simple'}$ has $O(k)$ leaves.
To bound the total number of nodes, we must consider the length of a path of degree $2$ nodes
in $\mathcal{T}_{simple'}$. If a node $\mu$ on this path has 
a split pair $\{u, v\}$ with an active vertex, say $u$, internal to its parent $\eta$ (implying $\eta$ does
not have $u$ in its split pair), then $u$ is exact for $\eta$ and we may associate $u$ with $\eta$.
Otherwise, consider two adjacent degree $2$ nodes to which this does not apply.
The parent may have an active vertex in its split pair not shared with its child, however then its child
cannot have an active vertex in its split pair or else it fits the above situation
(the active vertex in the child's split pair is internal to the parent).
Therefore this can only happen once
until the parent and child have at most one active vertex shared in their split pairs, and the procedure
of Lemma \ref{lem:tritreeprune} may be used to replace the parent node with the child.
It follows that we may associate each active vertex along this
path to a constant number of SPQR nodes.

Since the number of nodes of degree $\geq3$ in a tree is bounded by the number of leaves, the above shows
$\mathcal{T}_{simple'}$ has $O(k)$ nodes. We now consider the full tree $\mathcal{T}$. 
This tree adds inactive split components with a constant number of Q-node leaves as well as
Q-node leaves themselves to active split components.
By Lemma \ref{lem:nodereduction},
for any node $\mu$ in $\mathcal{T}_{simple'}$, we add at most $c_0 \cdot k$ extra children to $\mu$
in $\mathcal{T}$, where $k$ is the sum of active vertices exact to $\mu$ and the number of active
split components, which are children of $\mu$ in $\mathcal{T}_{simple'}$. The sum of the degrees of all
the vertices in $\mathcal{T}_{simple'}$ is $O(k)$, therefore the number of extra nodes in $\mathcal{T}$
is also $O(k)$. Therefore $\mathcal{T}$ has $O(k)$ nodes.

\end{proof}

We now consider applying Lemmas \ref{lem:hangingcomp2}, \ref{lem:nodereduction},
and \ref{lem:tritreeprune} on $G$ to produce a $3$-vertex equivalent graph $H$ in linear time.

\begin{lemma}
Given a bi-connected graph $G$ and $k$ active vertices $W$,
we may find the SPQR tree associated with $G$ and apply
the reduction rules given by 
Lemmas~\ref{lem:hangingcomp2},~\ref{lem:nodereduction}, and \ref{lem:tritreeprune}
until exhaustion in linear time.
From this can be constructed a graph $H$ $3$-vertex equivalent to $G$
with $O(k)$ vertices and edges.
\end{lemma}

\begin{proof}
The SPQR tree can be found in linear time \cite{BienstockM88}.
Lemma \ref{lem:manyevent} shows that continued application of
Lemmas~\ref{lem:hangingcomp2}, \ref{lem:nodereduction}, and \ref{lem:tritreeprune}
produces an SPQR tree $\mathcal{T}$ with $O(k)$ nodes. As each leaf
of $\mathcal{T}$ is an edge of the graph it represents, this shows the resulting graph
represented has $O(k)$ vertices and edges.

The rules given by Lemmas~\ref{lem:hangingcomp2}, \ref{lem:nodereduction}, and \ref{lem:tritreeprune}
can be applied until exhaustion in $O(m)$ time. First, we may apply Lemma \ref{lem:hangingcomp2}
by traversing each split component of $\mathcal{T}$ and applying the lemma when possible. Next,
we may apply Lemma \ref{lem:nodereduction} by a similar traversal. Finally, the rule of Lemma
\ref{lem:tritreeprune} can be applied by keeping a stack of the SPQR nodes of a depth-first traversal
of the SPQR tree $\mathcal{T}$. All three rules require a single traversal of $\mathcal{T}$ each and thus
can be done in $O(m)$ total time.
\end{proof}

\subsection{Constructing the Full Equivalent}
\label{subsec:trifull}
We now extend this trimming routine for bi-connected graphs
to arbitrary graphs.
By an argument similar to that at the start of Section \ref{sec:Two},
we may assume that $G$ is connected.
Therefore we need to work on the block-tree in a manner
similar to Section \ref{subsec:twovert}.
We first show that we may invoke the trimming routine from
Section \ref{subsec:spqrtrim} on each bi-connected component.

\begin{lemma}
\label{lem:bicompreplace}
Let $B$ be a bi-connected component in $G$ and $B'$ be a $3$-vertex equivalent
graph for $B$ where the active vertices consist of
all cut vertices and active vertices in $B$.
Replacing $B$ with $B'$ gives $H$ that's $3$-vertex
equivalent to $G$.
\end{lemma}

\begin{proof}
By symmetry of the setup it suffices to show that any
cut in $G$ has an equivalent cut in $H$.
Since $B$ is two-connected, any cut in $B$ that does not
remove at least $2$ vertices in $B$ results in all vertices of $B$
on the same side of the cut.
This means any vertex in $B$ that's not a cut-vertex in $G$
can be added without changing the cut.
This also means that removing these vertices has the
same effect in $G$ and $H$.

Therefore it suffices to consider cuts that remove $2$
vertices in $B$.
By assumption, there exists a cut in $H$ that results
in the same set of active vertices in $C$, and cut vertices
incident to $B$ on one side.
The structure of the block tree gives that the rest
of the graph is connected to $B$ via one of the cut
vertices, and the side that this cut vertex is on
dictates the side of the cut that part is on.
Therefore, the rest of the graph, and therefore the active vertices
not in $B$ will be partitioned the same way by this cut.
\end{proof}

It remains to reduce the block tree, which we do in a way
analogous to Section~\ref{subsec:twovert}.
However, the proofs need to be modified for separation
pairs instead of a single cut edge/vertex.

\begin{lemma}
\label{lem:hangingcomp}
A bi-connected component containing no active vertices
and incident to only one cut vertex can be removed (while
keeping the cut-vertex) to create a $3$-vertex equivalent graph.
\end{lemma}

\begin{proof}
Any cut in $G$ can be mapped over to $H$ by removing
the same set of vertices (minus the ones in this component).
Given a cut in $H$, removing the same set of vertices
in $G$ results in a cut with this component added to
one side.
As there are no active vertices in this cut, the cut
separates the active vertices in the same manner.
\end{proof}

Long paths of bi-connected components can be handled
in a way similar to Lemma \ref{lem:bitreeprune}.

\begin{lemma}
\label{lem:deg2comp}
Two bi-connected components $C_1$ and $C_2$
with no active vertices that share cut vertex $w$
and are only incident to one other cut vertex each,
$u$ and $v$ respectively, can be replaced by an
edge connecting $u$ and $v$, preserving $3$-vertex equivalence.
\end{lemma}

\begin{proof}
As there is a path from $u$ to $v$ through these
components and $w$, a minimum cut that separates $W$ in $G$ either removes
one of $u$, $v$, or $w$, or has $uv$ on the same side.
In the first case, removing $u$ or $v$ leads to the
same cut as none of them are active, while in the
second case the edge $uv$ does not cross the cut.
Since $uv$ is connected by an edge in $H$, in a cut in $H$
they're either on the same side, or one of $u$ or $v$ is removed.
In either case, an equivalent cut in $G$ can be formed
by assigning what remains of $C_1$ and $C_2$
to the same side as any remaining vertex.
\end{proof}

Our final construction can be obtained by applying these routines
to the block tree first, then Lemma~\ref{lem:bicompreplace}.

\begin{lemma}
Given a graph $G$ with $m$ edges and $k$ active vertices $W$,
a $3$-vertex equivalent of $G$ of size $O(k)$, $H$, can be constructed
in $O(m)$ time.
\end{lemma}

\begin{proof}
By a proof similar to Lemma \ref{lem:biequivalent}, applying
Lemmas \ref{lem:hangingcomp} and \ref{lem:deg2comp}
repeatedly leads to a $3$-vertex equivalent graph $H$ whose block-tree
contains $O(k)$ components and cut vertices.

The total number of cut vertices and active vertices
summed over all bi-connected components is $O(k)$.
Therefore, applying Lemma~\ref{lem:bicompreplace} on
each component gives the size bound.
\end{proof}

%!TEX root = Main.tex

\section{Omitted Proofs}
\label{sec:Appendix}

\begin{proof}[Proof of Lemma \ref{lem:Remove3EdgeCon}]
We prove by contradiction. Let $G$ be a graph with no nontrivial $3$-edge
connected component.
Suppose there exists two simple cycles
$a$ and $b$ in $G$ with more than one vertex, and thus at least one edge, in common.

Call the vertices in the first simple cycle $a_1, \ldots, a_n$ and the
second simple cycle $b_1, \ldots, b_m$, in order along the cycle.

Since these cycles are not the same, there must be some vertex not common to both cycles.
Without loss of generality, assume
(by flipping $a$ and $b$) that $b$ is not a subset of $a$,
and (by shifting $b$ cyclically)
that $b_1$ is only in $b$ and not $a$.

Now let $b_{first}$ be the first vertex after $b_1$
in $b$ that is common to both cycles, so
\begin{align}
first \defeq \min_{i} b_i \in a.
\end{align}
and let $b_{last}$ be the last vertex in $b$
common to both cycles
\begin{align}
last \defeq \max_{i} b_i \in a.
\end{align}
The assumption that these two cycles have more than
$1$ vertex in common means that
\begin{align}
first < last.
\end{align}

We claim $b_{first}$ and $b_{last}$ are 3-edge connected.

We show this by constructing three edge-disjoint paths connecting $b_{first}$ and $b_{last}$.
Since both $b_{first}$ and $b_{last}$ occur in $a$,
we may take the two paths formed by cycle $a$
connecting $b_{first}$ and $b_{last}$, which are clearly edge-disjoint.

By construction, vertices
\begin{align}
b_{last+1}, \ldots, b_m, b_1, \ldots, b_{first-1}
\end{align}
are not shared with $a$.
Thus they form a third edge-disjoint path connecting
$b_{first}$ and $b_{last}$, and so the claim follows.
Therefore, a graph with no $3$-edge connected vertices,
and thus no nontrivial $3$-edge connected component
has the property that two simple cycles have at most
one vertex in common.
\end{proof}

\end{document}